\documentclass[a4paper,11pt]{amsart} 
\usepackage{amsmath}
\usepackage{amsfonts}
\usepackage{amssymb}
\usepackage{amsthm}
\usepackage{newlfont}
\usepackage{color}
\usepackage[pdftex]{hyperref}
\usepackage{subfig}
\usepackage{bbm}

\usepackage{mathrsfs}
\usepackage{palatino}  %palatcm - a missing package
\usepackage{fancyhdr}
\usepackage{graphicx}
\usepackage{geometry}
\usepackage{psfrag}

\newcommand{\remark}{\noindent{\it Remark. }}

 \let\be=\beta  

  \let\ga=\gamma 
 \let\la=\lambda  
 \let\vs=\varsigma

\newcommand{\caA}{{\mathcal A}}
\newcommand{\caB}{{\mathcal B}}
\newcommand{\caC}{{\mathcal C}}

\newcommand{\caF}{{\mathcal F}}
\newcommand{\caG}{{\mathcal G}}
\newcommand{\caH}{{\mathcal H}}

\newcommand{\caK}{{\mathcal K}}
\newcommand{\caL}{{\mathcal L}}
\newcommand{\caM}{{\mathcal M}}

\newcommand{\caO}{{\mathcal O}}

\newcommand{\caR}{{\mathcal R}}
\newcommand{\caS}{{\mathcal S}}

\newcommand{\caW}{{\mathcal W}}
\newcommand{\caX}{{\mathcal X}}
\newcommand{\caY}{{\mathcal Y}}
\newcommand{\caZ}{{\mathcal Z}}

\newcommand{\bbC}{{\mathbb C}}

\newcommand{\bbE}{{\mathbb E}}

\newcommand{\bbN}{{\mathbb N}}

\newcommand{\bbP}{{\mathbb P}}

\newcommand{\bbR}{{\mathbb R}}

\newcommand{\bbZ}{{\mathbb Z}}

\newcommand{\mfa}{\mathfrak{a}}
\newcommand{\mfb}{\mathfrak{b}}

\newcommand{\oX}{\overline{X}^{\lambda}}
\newcommand{\rX}{X^{\lambda}}
\newcommand{\oV}{\overline{V}^{\lambda}}
\newcommand{\rV}{V^{\lambda}}
\newcommand{\upl}{^{\lambda}}

\newcommand{\cha}{\mathrm{cha}}

\newcommand{\ie}{{\it i.e.\/} }

\newcommand{\iu}{\mathrm{i}}

\newcommand{\str}{^{*}}
\newcommand{\ep}[1]{\mathrm{e}^{#1}}
\newcommand{\hilb}{\mathcal{H}}
\newcommand{\dd}{\mathrm{d}}

\newcommand{\Tr}{\mathrm{Tr}}

\newcommand{\streep}{ |}
\newcommand{\norm}{ ||}

\newcommand{\beq}{\begin{equation}}
\newcommand{\eeq}{\end{equation}}
\newcommand{\beqn}{\begin{equation*}}
\newcommand{\eeqn}{\end{equation*}}
\newcommand{\baq}{ \begin{eqnarray} }
\newcommand{\eaq}{ \end{eqnarray} }
\newcommand{\mean}[1]{\mathbb{E}\left(#1\right)}

\newtheorem{thm}{Theorem}
\newtheorem{cor}[thm]{Corollary}
\newtheorem{prop}[thm]{Proposition}
\newtheorem{lma}[thm]{Lemma}
\newtheorem{defin}{Definition}

\begin{document}

\title[Disordered wires]{Disordered quantum wires: microscopic origins of the DMPK theory and Ohm's law}

\author[S. Bachmann]{Sven Bachmann}
\address{Department of Mathematics \\
University of California, Davis \\
One Shields Ave \\
Davis, CA 95616, USA}
\email{svenbac@math.ucdavis.edu}

\author[M. Butz]{Maximilian Butz}
\address{Fakult\"at f\"ur Mathematik \\ Technische Universit\"at M\"unchen \\
Boltzmannstr.~3 \\ 
85748 Garching, Germany}
\email{butz@ma.tum.de}

\author[W. de Roeck]{Wojciech de Roeck	}
\address{Physikalisches Institut \\ 
Universit\"at zu K\"oln \\
Z\"ulpicher Str. 77 \\
50937 K\"oln, Germany\\
On leave from University of Heidelberg}
%\address{Institut f\"ur Theoretische Physik \\
%Universit\"at Heidelberg \\
%Philosophenweg~16 \\
%69120 Heidelberg, Germany}
\email{wderoeck@thphys.uni-heidelberg.de}

\date{\today }

\begin{abstract}
We study the electronic transport properties of the Anderson model on a strip, modeling a quasi one-dimensional disordered quantum wire. In the literature, the standard description of such wires is via random matrix theory (RMT). Our objective is to firmly relate this theory to a microscopic model.  We correct and extend previous work \cite{BdR} on the same topic. In particular, we obtain through a physically motivated scaling limit an ensemble of random matrices that is close to, but not identical to the standard transfer matrix ensembles (sometimes called TOE, TUE), corresponding to the Dyson symmetry classes $\be=1,2$.
In the $\beta=2$  class, the resulting conductance is the same as the one from the ideal ensemble, i.e.\ from TUE. In the $\be=1$ class, we find a deviation from TOE. It remains to be seen whether or not this deviation vanishes in a thick-wire limit, which is the experimentally relevant regime. For the ideal ensembles, we also prove Ohm's law for all symmetry classes, making  mathematically precise a moment expansion by Mello and Stone \cite{mellostone}. This proof bypasses the explicit but intricate solution methods that underlie most previous results.
\end{abstract}

\maketitle

%%%%%%%%%%%%%%%%%%%%%%%%%%%%%%%%%%%%%%%%%%%%%%%%%%%%%%%%%%%%

\section{Introduction}\label{sec: introduction}

We start  below with a brief introduction to the physics of quasi one-dimensional quantum wires. In Section \ref{sec: models and results}, we sketch the scope of this paper and its relation to previous works.  The appropriate random matrix theory is discussed in Section \ref{sec: DMPK}. Our microscopic model, convergence results and proofs are presented in Sections \ref{sec: micro} and \ref{sec: proofs}. 
%  The remark  on local eigenvalue statistics stands, from the point of view of methodology,  apart from the rest of the paper and we address it in Section \ref{sec: GUE}.  Although it uses some definitions from Section \ref{sec: micro}, that part can easily be read on its own.

\subsection{Phenomenology}\label{sec: phenomenology}

Without yet introducing a concrete mathematical framework, we present the basic physics setup of quantum wires and try to elucidate the questions of charge transport and conductance fluctuations. We rely heavily on the excellent review \cite{beenakkerreview}.  

Disordered quantum wires are quasi one-dimensional pieces of dirty (disordered) conductor.  The wire has a physical length $L$, which is most conveniently expressed in units of the \emph{mean free path} $\ell$ so that we shall use $s:= L/\ell$. In a microscopic model where the parameter $\la\geq0$ measures the strength of the disorder $\ell\sim \la^{-2}$. The width $W$ of the wire is expressed by an integer $N$ that corresponds to the number of different modes that `fit' in the wire. Physically, $N \sim W/\la_F$ with $\la_F$ the Fermi wavelength of the electrons sent through the wire, which is in turn determined by the energy of those incoming electrons.

For a purely one-dimensional wire, $N=1$, it is well-known that an electron travelling through the wire gets localized with localization length of the order of the mean free path $\ell$, hence $s \sim 1$.  However, the localization length increases with $N$ (roughly as $s \sim N$, at least in the weak disorder limit $\la\to 0$) and we can ask how the system behaves for $s \ll N$, before localization sets in. 
There, one can distinguish the \emph{ballistic} regime $s  \leq 1$, where incoming electrons did not yet get scattered by the impurities,  and the most interesting \emph{diffusive} regime characterized by 
\beq \label{def: diffusive regime}
1 \ll s, \qquad   s/ N \ll 1\,.
\eeq
One of the fascinating aspects of this regime is the phenomenon of universal conductance fluctuations (UCF) first discussed in~\cite{LeSt85}.  Let $g= g(s,N, \la)$ be the conductance of the wire, expressed in units of the conductance quantum $2e^2/\hbar$. It is a random quantity due to the disorder.  In the thick wire limit, its disorder average $\bbE(g)$, is roughly given by 
\beq
 \bbE (g)    \sim   \left\{  \begin{array}{lllll}   N/s     &  \qquad &    1 \ll s, \quad   s/ N \ll 1   &  \qquad & \textrm{(Ohm's law)}  \\[2mm]      \exp{\{  -  s /N \} } &  \qquad &   s  > N &  \qquad & \textrm{(localization)}        \end{array} \right. \label{eq: rough conjecture}
\eeq
Furthermore, universal conductance fluctuations mean that, in the diffusive regime defined by \eqref{def: diffusive regime}, 
\beq
  \mathrm{Var}(g)  = 2/(15\beta),
\eeq
independently of the microscopic details of the wire, or its length and width. The only parameter that remains in this regime is the symmetry index $\be$ that refers to Dyson's symmetry classes. 

We emphasize that these phenomena should emerge in a large $N$ limit only. On the other hand, $N$ cannot be too large because then we enter the regime of two-dimensional localization, at least if we assume that the wire has one transverse dimension. However, even if the transverse dimension is higher, the reasoning breaks down as soon as $W > \ell$.  It is therefore important to take a weak-disorder limit first, $\la\to 0$, which also means that the wire's microscopic length $L=\lambda^{-2}s$ diverges.  Below, we try to distill some precise conjectures that are generally accepted.  From the mathematical perspective, they can be partially proven if one accepts RMT as a starting point (see Section \ref{sec: DMPK}), but open if one starts from a more realistic model, as the one treated in Section \ref{sec: micro} of this article.\\

\noindent\textbf{Conjecture 1} (Ohm's law)
\beq
%\lim_{\substack{N \nearrow \infty \\  s/N =z \, \textrm{fixed}} }
  \lim_{N \to \infty}  \quad   \lim_{\la \to 0}   \quad   \frac{1}{N}\bbE (g)  =  \frac{1}{s}  + o(1/s), \qquad  s \to \infty.
\eeq

\noindent\textbf{Conjecture 2} (Universal conductance fluctuations)
\beq
  \quad \lim_{N \to \infty}  \quad   \lim_{\la \to 0}   \quad     \mathrm{Var}(g)  =  \frac{2}{15 \beta} +o(1), \qquad  s \to \infty.
\eeq

We stress here that these conjectures reflect the minimum of what should be true according to the literature, and that the underlying heuristics is quite involved. The present paper does partially settle theses conjectures starting from a microscopic model but with an additional scaling limit, as will be explained in the next section. 

\subsection{Setup, goals, and results}\label{sec: models and results}

The standard approach to disordered quantum wires is to model the transfer matrix of such a wire by an appropriate ensemble of random matrices. The matrices under consideration belong to a subgroup of pseudo-unitary matrices. Following~\cite{TOETUE} we shall call their ensembles TOE, TUE and TSE in analogy to the better known ensembles of Hamiltonians, the hermitian GOE, GUE and GSE, or the circular ensembles of unitaries: COE, CUE and CSE. In fact, ensembles of transfer matrices come with a real positive parameter, called $s$ above  and physically corresponding to the length of the wire. They are therefore more complicated, but also more interesting: in particular the  parameter $s$ tunes a localization-delocalization transition. This can be observed for example in the Fokker-Planck equation describing the $s$-dependence of the conductance, which is the equation usually referred to as the DMPK equation.

The natural question arises whether the RMT ensembles allow for a verification of the conjectures mentioned at the end of the previous section, with the proviso that the $\la\to 0$ should be omitted as the RMT assumes weak coupling from the start. In the physics literature, there is overwhelming evidence for an affirmative answer, and the conjectures have been verified in \cite{chalkermacedo, beenakkerrejaei, caselledistribution, zirnbauersuperfourier}. In that perspective, we shall here give a rigorous proof of Ohm's law based on a moment argument of~\cite{mellostone}, thereby confirming Conjecture~1 for the TOE, TUE and TSE.

The ultimate goal of our work is a derivation of the conjectures from a more realistic model of the wire, i.e.\ from a `reasonable' microscopic Hamiltonian, namely the Anderson model on a tube of width $N$ with a disordered region of length $L$ and disorder strength $\lambda$. First, we need to be in the weak coupling regime $\la\to0$, and therefore $L=\lambda^{-2}s\to\infty$. This first scaling limit yields a random matrix ensemble $\caG(s)$, see Proposition \ref{Prop:scalingNIsotropic}. For the conjectures to hold, a second scaling is certainly necessary, namely that of a broad wire, $N\to\infty$. At the time of writing, the validity of the conjectures in this scaling regime remains an open question. However, if we consider an additional scaling limit in which the transversal hopping in the wire is small compared to the longitudinal hopping, see Theorem~\ref{thm:convA}, we obtain, instead of the ensemble $\caG(s)$, a new transfer matrix ensemble $\caA(s)$ that is very close to the ideal ensemble. In fact, for $\be=2$, the conductance calculated from that ensemble is the same as that calculated from the TUE. Since we proved Ohm's law for the $N\to\infty$ limit of the random matrix ensemble in the first place, this provides a proof of Conjecture~1 in a weaker sense for $\beta=2$. In Section~\ref{sec: discussion}, we comment on the ensemble $\caA(s)$, pointing out to how and why it fails to satisfy all the symmetry properties of the ideal ensembles. 

This article is to a large extent based on a previous paper \cite{BdR} by two of us, which appeared on the arXiv shortly after and independently of \cite{viragvalkoexplosion}. Despite their similarity these two articles stressed different aspects of the resulting transfer matrix evolutions. However, \cite{BdR} contained an error, as pointed out by the second author of the present paper, and the symmetry properties of the model were not consistently treated. In this article, which supersedes~\cite{BdR}, we first extend the setup by constructing models for both $\be=1$ and $\be=2$ symmetry classes\footnote{The physically most natural way to discuss $\be=4$ as well would be to consider electrons with spin, which we chose not to do for reasons of simplicity}.
 Moreover, we incorporate technical improvements (among other things borrowing some terminology from \cite{viragvalkoexplosion}), mostly concerning the statement of the joint scaling limit in Theorem \ref{thm:convA}. Finally, we study the convergence as $N\to\infty$ of a hierarchy of equations for the moments of the conductance introduced by~\cite{mellostone}. We prove that the limit satisfies Ohm's law, see Theorem~\ref{thm:Ohm}.

%%%%%%%%%%%%%%%%%%%%%%%%%%%%%%%%%%%%%%%%%%%%%%%%%%%%%%%%%%%%

\section{Random matrix theory: the DMPK equation}\label{sec: DMPK}

%%%%%%%%%%%%%%%%%%%%%%%%%%%%%%%%%%%%%%%%%%%%%%%%%%%%%%%%%%%%

Transport properties of a quasi one-dimensional system are most conveniently approached through its scattering matrix, or equivalently its transfer matrix. In this section we shall consider these objects as the fundamental quantities of the theory, understand what symmetries imply on their general structure and derive a stochastic differential equation describing their behavior as a function of the length of the disordered wires, based on an isotropy assumption, also called `(local) maximal entropy' Ansatz. In particular, we do not assume that the transfer matrices here arise from some sort of microscopic Hamiltonian dynamics.

%%%%%%%%%

\subsection{Transfer matrices and symmetries}\label{sub: Transfer}

Heuristically speaking the transfer matrix of a quasi one-dimensional wire maps free waves on the far right of the sample to free waves on the far left of it. Although this picture is physically meaningful, we shall only refer to it explicitly in Section~\ref{sec: micro} and keep an abstract point of view here. We first fix a preferred basis in $\bbC^{2N}$ and make the following definition.
\begin{defin} \label{Def:Transfer}
A \emph{transfer matrix} for a wire of width $N$ is a $2N\times 2N$ pseudo-unitary matrix,
\begin{equation}\label{PseudoUnitary}
\caM\str \Sigma_z \caM = \Sigma_z\,,\qquad\text{where}\qquad 
\Sigma_z = \begin{pmatrix} 1 & 0 \\ 0 & -1\end{pmatrix}\,.
\end{equation}
Furthermore, a transfer matrix $\caM$ is \emph{time reversal invariant} if
\begin{equation}\label{M_TRI}
\Sigma_x \caM \Sigma_x = \overline{\caM}\,,\qquad\text{where}\qquad 
\Sigma_x = \begin{pmatrix} 0 & 1 \\ 1 & 0 \end{pmatrix}\,,
\end{equation}
and $\overline{\caM}$ is the complex conjugate of $\caM$.
\end{defin}
Transfer matrices have a simple multiplicative composition rule. If $\caM_1$ and $\caM_2$ are transfer matrices for two wires, then $\caM_2\caM_1$ is the transfer matrix for the composite system obtained from gluing the two pieces in series.

In view of~(\ref{PseudoUnitary}, \ref{M_TRI}), it is natural to write a transfer matrix in block form
\begin{equation*}
\caM = \begin{pmatrix} \caM_{++} & \caM_{+-} \\ \caM_{-+} & \caM_{--} \end{pmatrix}\,.
\end{equation*}
Combining~(\ref{PseudoUnitary}) and the singular value decompositions $\caM_{++} = U_+ S_+ V_+$ and $\caM_{--} = U_- S_- V_-$ of the diagonal blocks, we obtain the following factorization
\begin{equation}\label{MSdecomp}
\caM = \begin{pmatrix} U_+ & 0 \\ 0 &  U_- \end{pmatrix} \begin{pmatrix} S & (S^2-1)^{1/2} \\ (S^2-1)^{1/2} &  S \end{pmatrix} \begin{pmatrix} V_+ & 0 \\ 0 &  V_- \end{pmatrix}\,.
\end{equation}
where $S=S_+=S_-$. If, moreover, time reversal invariance is imposed, then $U_- = \overline{U_+}$ and $V_- = \overline{V_+}$. 

Let $r$ be the reflection and $t$ be the transmission matrices, defined through
\begin{equation*}
\caM \begin{pmatrix} 1 \\ r \end{pmatrix} = \begin{pmatrix} t \\ 0 \end{pmatrix} 
\end{equation*}
In particular,
\begin{equation*}
t = \caM_{++}-\caM_{+-}\caM_{--}^{-1}\caM_{-+}= \left(\caM_{++}\str\right)^{-1}\,,
\end{equation*}
where we used~(\ref{PseudoUnitary}) in the second equality. The so-called transmission eigenvalues $(T_k)_{k=1}^N$ are defined as the eigenvalues of the matrix $t\str t$, i.e.\ of $(\caM_{++}\str\caM_{++})^{-1}$. Hence, the transmission eigenvalues are also the inverses of the squares of the singular values contained in $S$. Let $T = S^{-2}$ be the diagonal matrix of transmission eigenvalues. Many  transport properties of the disordered wire can be expressed as functions thereof. In particular, the conductance $g$ is given in units of $2e^2 / \hbar$ by the Landauer-B\"uttiker formula~\cite{LB},
\begin{equation*}
g:= \sum_{i=1}^N T_i = \Tr\, T = \Tr\, t\str t\,,
\end{equation*}
a formula that we accept here as a definition of $g$.

%%%%%%%%%

\subsection{The DMPK Theory}\label{sub: DMPK}

The DMPK theory introduced by~\cite{D} and independently by~\cite{MPK} is an evolution equation for the transfer matrix $\caM(r,s)$ of a wire on $[r,s]$. By the composition rule, for any $s_1\leq s_2$,
\beq
\caM(0,s_2) =\caM(s_1,s_2) \caM(0,s_1)\,,
\eeq
with $\caM(s,s)=1$. The first crucial idea is to take $s_2-s_1$ infinitesimal and write 
\beq
\caM(s,s+\dd s) \sim 1 + \dd  \caL(s)
\eeq
such that $\dd  \caL(s)$ is independent of $\caM(s)$ and contains only diffusive terms but no drift.  Mathematically, this translates into the assumption that $\caM(s)$ satisfies an It\^o stochastic differential equation  (SDE)
\begin{equation}
\label{DMPK}
 \begin{split}
  \dd  \caM(s) &= \dd  \caL(s) \caM(s)\,,\\
\caM(0)&=1\,,
 \end{split}
\end{equation}
where $\caL(s)$ is a matrix valued Brownian motion and $\caL(0) = 0$.
\begin{lma}
Let $\caM(s)$ be a solution of the SDE~(\ref{DMPK}). Assume that
\begin{align}
\dd  \caL\str\Sigma_z + \Sigma_z \dd  \caL &= 0\,, \label{PseudoU1}\\
\dd \caL\str\Sigma_z \dd  \caL &= 0\,. \label{PseudoU2}
\end{align}
Then $\caM(s)$ is pseudo unitary, eq.~(\ref{PseudoUnitary}). If moreover
\begin{equation}\label{L_TRI}
\Sigma_x \dd  \caL \Sigma_x = \overline{\dd \caL}\,,
\end{equation}
then $\caM(s)$ is also time reversal invariant, eq.~(\ref{M_TRI}).
\end{lma}
\begin{proof}
For the first part, we take the differential of~(\ref{PseudoUnitary}), use~(\ref{DMPK}) and It\^o calculus to obtain
\begin{equation*}
\caM\str\left(\dd \caL\str\Sigma_z + \Sigma_z \dd \caL + \dd \caL\str\Sigma_z \dd \caL\right)\caM = 0\,,
\end{equation*}
which holds if and only if both~(\ref{PseudoU1}) and~(\ref{PseudoU2}) hold as $\caM$ is nonsingular. Similarly, the differential of~(\ref{M_TRI}) immediately yields~(\ref{L_TRI}).
\end{proof}

Secondly, the DMPK theory prescribes a particular invariance of the distribution of $\dd \caL(s)$.  The law of the increments $\dd \caL(s)$ shall be independent of $s$ and maximally isotropic in the sense that
\begin{equation}\label{dLEntropy}
\caW\str \dd \caL\, \caW \mathop{=}\limits^{d} \dd \caL\qquad\text{for any unitary}\qquad \caW = \begin{pmatrix} W_+ & 0 \\ 0 & W_- \end{pmatrix}\,.
\end{equation}
The unitary blocks $W_\pm$ are independent of each other if $\caM$ does not exhibit any symmetry, whereas $W_- = \overline{W_+}$ if time reversal symmetry is imposed. For notational simplicity, we cast $\caL$ in block form,
\begin{equation*}
\caL(s) = \begin{pmatrix} \mathfrak{a}(s) & \mathfrak{b}(s) \\ \mathfrak{b}(s)\str & \mathfrak{a}'(s) \end{pmatrix}
\end{equation*}
where $\mathfrak{a}(s),\mathfrak{a}'(s), \mathfrak{b}(s)$ are independent local martingales, with $\mathfrak{a}(s) = -\mathfrak{a}(s)\str$, similarly for $\mathfrak{a}'(s)$, and
\begin{equation}\label{BlockU}
\dd \mathfrak{a}\str \dd \mathfrak{a} = \dd \mathfrak{b}\dd \mathfrak{b}\str = \dd \mathfrak{a}'{\dd \mathfrak{a}'}\str\,.
\end{equation}
The isotropy assumption~(\ref{dLEntropy}) reduces to invariance conditions on the blocks. First, 
\begin{equation} \label{DMPKa}
\mathfrak{a}_{ij}(s) = \begin{cases}
1 / \sqrt{2N}  \cdot (B^R_{ij}(s) + \iu B^I_{ij}(s)) & 1 \leq i < j \leq N \\ 
\iu / \sqrt{N} \cdot B^I_{ii}(s) & i=j \\ 
-\overline{\mathfrak{a}_{ji}(s)} & \text{otherwise}
\end{cases}\,,
\end{equation}
where $B^R$ and $B^I$ are independent real standard Brownian motions, and similarly but independently for $\mathfrak{a}'(s)$. Secondly, 
\begin{equation} \label{DMPKb}
\mathfrak{b}_{ij}(s) = 1 / \sqrt{2N}  \cdot (\tilde B^R_{ij}(s) + \iu \tilde B^I_{ij}(s))\,,\quad\text{for all }i,j\,.
\end{equation}
Note that the relative normalization of $\mathfrak{a}(s)$ and $\mathfrak{b}(s)$ are fixed by pseudounitarity, i.e.\ (\ref{BlockU}). In the time reversal invariant case, the matrix $\mathfrak{a}(s)$ does not change, but
\begin{equation*}
\mathfrak{a}'(s) = \overline{\mathfrak{a}(s)}\,,
\end{equation*}
and $\mathfrak{b}(s)$ becomes symmetric, $\mathfrak{b}(s)\str = \overline{\mathfrak{b}(s)}$ with real and imaginary parts orthogonally invariant, namely
\begin{equation} \label{DMPKbTRI}
\mathfrak{b}_{ij}(s) = \begin{cases}
1 / \sqrt{2(N+1)}  \cdot (B^R_{ij}(s) + \iu B^I_{ij}(s)) & 1 \leq i < j \leq N \\ 
1/\sqrt{N+1} \cdot (B^R_{ii}(s) + \iu B^I_{ii}(s)) & i=j \\ 
\mathfrak{b}_{ji}(s) & \text{otherwise}
\end{cases}\,.
\end{equation}
Here again, the relative factor $\sqrt{N/(N+1)}$ is imposed by~(\ref{BlockU}).

From a physical point of view, the DMPK theory's interest lies in its predictions for the statistics of  the transmission eigenvalues. Indeed, the unitary invariance of the increments $\dd  \caL$ implies that the set of  $T_k$ satisfies an autonomous equation, which can be formally derived by It\^o calculus from the matrix SDE~(\ref{DMPK}):
\begin{equation}
\begin{split}
\label{dmpkeq}
 \dd T_k(s)&=v_k(T(s))\dd s+D_k(T(s))\dd B_k(s),\\
T_k(0)&=1,
\end{split}
\end{equation}
for all $k=1,...,N$. The Brownian motions $B_k$ are independent, and the drift and diffusion coefficients are given explicitly by
\begin{equation*}
\label{driftdiff}
 \begin{split}
  v_k&=-T_k+\frac{2T_k}{\beta N+2-\beta}\left(1-T_k+\frac{\beta}{2}\sum_{j\neq k}\frac{T_k+T_j-2T_kT_j}{T_k-T_j}\right)\,,\\
D_k&=\sqrt{4\frac{T_k^2(1-T_k)}{\beta N+2-\beta}}.
 \end{split}
\end{equation*}
The first term in the drift $v_k$ contracts all transmission eigenvalues towards $0$ as the length $s$ of the wire increases. However, and similarly to Dyson's Brownian motion, the drift also contains repulsion terms originating from second order perturbation theory. As a consequence, the eigenvalues $T_k$ `try to avoid' degeneracy. What makes a naive derivation formal is that It\^o's formula is only applicable if the denominator $T_k-T_j$ never becomes singular, \ie{}$\caM^*_{++}(s)\caM_{++}(s)$ never has degenerate eigenvalues. This is a nontrivial property for $s>0$, and even more so as $s\to0^+$ since~(\ref{DMPK}) starts with the completely degenerate $\caM(0)=1$. Both issues can however be tackled and the SDE~(\ref{dmpkeq}) has a unique weak and strong solution, see~\cite{Butzpaper}.

Finally, let us comment on some deeper principles underlying the process $\caM$ and the resulting DMPK equation. 
The maximal isotropy assumption that was used above, can be derived from a simple `maximal entropy assumption' on the set of infinitesimal transfer matrices $1+\dd \caL(s)$ that have a fixed 'scattering strength' $\sum_k T_k$.
Alternatively, as remarked by \cite{huffmann}, one can also guess the DMPK equation from geometric considerations, since it is the radial part of the canonical Brownian motion on a certain symmetric space. The reduction from Lie group to symmetric space is obtained by identifying certain transfer matrices that, in particular, have the same transmission eigenvalues.
This geometric approach was very fruitful. For example, in  \cite{brouwernonuniversality} it was shown how it naturally explains the appearance of non-universal conductance properties in wires with off-diagonal disorder.  

%%%%%%%%%

\subsection{Ohm's law}\label{sub: Ohm}

In the context of the DMPK theory, a treatment, or even proof, of  the conjectures mentioned in the introduction is possible, as already indicated in Section \ref{sec: models and results}.   The existing approaches rely on explicit calculations and are quite intricate.
Nevertheless, if one is solely after Ohm's law (and not the universal conductance fluctuations (UCF)),  there is an appealing and compact approach by \cite{mellostone}. Below we present a rigorous version of this approach.
% The transmission eigenvalues $T_k$ are uniformly bounded by $1$, so that the conductance and its moments satisfy the simple a priori bound
%%
%\begin{equation*}
%\frac{g_N^p}{N^p} \leq 1
%\end{equation*}
%%
%for all $p\geq 1$. 

The following theorem shows that in the large $N$ limit, the rescaled moments of the conductance have an Ohmic behavior. In particular, Conjecture~1 holds for the TOE, TUE and TSE. We note that the symmetry index $\beta$ drops out in that particular scaling.
\begin{thm}[Ohm's law] \label{thm:Ohm}
Let $\left(T_k(s)\right)_{k=1}^N$ be the solution of the DMPK process~(\ref{dmpkeq}), and let
\begin{equation*}
g_N(s) = \sum_{k=1}^N T_k(s)\,.
\end{equation*}
Then for all $p\geq 1$ and $T>0$, 
\begin{equation} \label{eq:Ohm}
\lim_{N\to\infty}\frac{\bbE(g^p_N(s))}{N^p} = \frac{1}{(1+s)^p}
\end{equation}
uniformly for $s\in[0,T]$.
\end{thm}
The proof goes through a sequence of lemmas. 
\begin{lma}\label{lma:hierarchy}
Let $g_N^{(j)} = \sum_k T_k^j$ for $j>1$. Then, for any $p\geq 1$, 
\begin{equation} \label{hierarchy}
\frac{\dd }{\dd s}\bbE(g_N^p) = -p\gamma_N(\beta)\left[
\bbE(g_N^{p+1}) - \left(1-\frac{2}{\beta}\right)\bbE(g_N^{p-1} g_N^{(2)}) -\frac{2(p-1)}{\beta}\bbE(g_N^{p-2} (g_N^{(2)}-g_N^{(3)}))
\right].
\end{equation}
where
\begin{equation*}
\gamma_N(\beta) = \frac{\beta}{\beta N +2 -\beta}.
\end{equation*}
\end{lma}
\begin{proof}
It\^o's formula yields
\begin{equation*}
\dd (g_N^p) = \sum_{i_1,\ldots,i_p = 1}^N \left[ 
\sum_{j=1}^p T_{i_1}\cdots dT_{i_j}\cdots T_{i_p} + \sum_{j\neq k=1}^p T_{i_1}\cdots dT_{i_j}\cdots dT_{i_k}\cdots T_{i_p}
\right]
\end{equation*}
so that
\begin{equation} \label{eq: unpolished hierarchy}
\frac{\dd }{\dd s}\bbE(g_N^p) = p\bbE\left(g_N^{p-1}\sum_k v_k\right) + \frac{p(p-1)}{2}\bbE\left(g_N^{p-2}\sum_k D_k^2\right)\,.
\end{equation}
where we used the DMPK equation~(\ref{dmpkeq}). In order to rewrite the right-hand side, note the simple identity
\begin{equation*}
\sum_k T_k \sum_{j\neq k}\frac{T_k+T_j-2T_kT_j}{T_k-T_j} = \sum_k\sum_{j < k}(T_k+T_j-2T_kT_j) = (N-1) g_N - g_N^2 + \sum_k T_k^2
\end{equation*}
Therefore,
\begin{align}
\sum_k v_k &= -g_N + \frac{2}{\beta N +2 -\beta}\left[g_N - \sum_k T_k^2 + (\beta/2)\left((N-1) g_N - g_N^2 + \sum_k T_k^2\right) \right] \nonumber\\
&=-\gamma_N(\beta)\left( g_N^2 - (1-2/\beta)g_N^{(2)}\right)\,,  \label{sumVk}
\end{align} 
and the lemma follows upon substituting this in \eqref{eq: unpolished hierarchy}.
\end{proof}

Let us now consider
\begin{equation*}
\Psi_{N}(p,s) :=\frac{\bbE(g^p_N(s))}{N^p}\,,
\end{equation*}
The following properties are immediate from the definition and the differential equation \eqref{hierarchy}
\begin{enumerate}
%\label{eq: properties of finite sequence}
\item  $\Psi_{N}(p,0)=1$.
\item  $\streep \Psi_{N}(p,s) \streep \leq 1$.
\item The function $s \mapsto \Psi_{N}(p,s)$ is continuously differentiable and $ \streep \frac{\partial \Psi_{N}}{\partial s}(p,s) \streep  \leq c(p) < \infty$.\end{enumerate}
We consider the Banach space $\caL=\caC([0,T], \bbR)$ for some $T>0$, equipped with the supremum norm. Let  $\caL_p, p=1,2,\ldots$ be copies of $\caL$ and define  the cartesian product $\caK= \mathop{\times}\limits_{p=1}^{\infty} \caL_p$,  equipped with the product topology.

\begin{lma}\label{lem: compactness}  There is an increasing sequence $N_n, n \in \bbN$ and an element $\Psi \in \caK$ such that for each $p$, $
\Psi_{N_n}(p,\cdot) \to \Psi(p,\cdot) $ in $\caL$, as $n \to \infty$.
\end{lma}
\begin{proof}
For $p=1,2,\ldots$, consider the sets $$ \caS_p=  \left\{  \Psi_{N}(p,\cdot), N \in \bbN\right\}  \subset \caL_p.$$
By the properties $\mathrm{i,ii,iii}$ above and Arzela-Ascoli's theorem, each of these sets is sequentially compact (s.c.). Countable products of s.c.\ sets are s.c.\ in the product topology (sequential Tychonov's theorem), hence $\times_p\caS_p \subset \caK$ is s.c.\  Therefore, the sequence $\Psi_N \subset \times_p\caS_p$  has a convergent subsequence.  Since convergence in the product topology implies convergence for any $p$,  the lemma follows.
\end{proof}

\begin{lma}  Any limit point $\Psi \in \caK$ as in Lemma \ref{lem: compactness}  satisfies 
\beq\label{eq: inf hierarchy}
\Psi(p,s_2)-\Psi(p,s_1) = - p \int_{s_1}^{s_2} \dd s \Psi (p+1,s), \qquad  \Psi(p,0) = 1
\eeq
\end{lma}
\begin{proof}
 The equation  \eqref{hierarchy} is rewritten as
\begin{equation} \label{eq: new form hierarchy}
\Psi_N(p,s_2)-\Psi_N(p,s_1)= -p  (1-r_1(N)) \int_{s_1}^{s_2} \dd s \left[
\Psi_N(p+1,s) - r_2(N,s,p) \right]
\end{equation}
where $r_2$ is the sum of the second and third term between square brackets in \eqref{hierarchy} and
\begin{align*}
r_1(N) &= 1-N\gamma_{N}  = \caO(1/N), \\
r_2(N,s,p) &= \left(1-\frac{2}{\beta}\right)\frac{\bbE(g_N^{p-1} g_N^{(2)})}{N^{p+1}} + \frac{2(p-1)}{\beta}\frac{\bbE(g_N^{p-2} (g_N^{(2)}-g_N^{(3)})}{N^{p+1}} = \caO(1/N)
\end{align*}
with the bounds $\caO(1/N)$ uniform in $s$ but not necessarily in $p$. 
The lemma follows by considering \eqref{eq: new form hierarchy}, for fixed $p$, along the sequence ${N_n}$.
\end{proof}

Sloppily put, the above two lemmas show that $\Psi_N(p,\cdot)$ converges to a solution of the `limiting hierarchy of equations' \eqref{eq: inf hierarchy}.   It remains to prove that the limiting hierarchy has a unique solution, namely the right-hand side of~\eqref{eq:Ohm}. 
Therefore, the proof of Theorem~\ref{thm:Ohm} is completed by the next lemma. 

\begin{lma} \begin{equation*}
\Psi (p,s): = \frac{1}{(1+s)^p}
\end{equation*}
is the unique element in $\caK$ that satisfies \eqref{eq: inf hierarchy} and $\sup_p \norm \Psi(p,\cdot)\norm_\infty \leq 1$.
\end{lma}
\begin{proof}
We proceed by induction in the interval length $T$. Assume that the claim is proven for $T\geq 0$ (for $T=0$ it is trivial).  Take then 
$0 \leq s_1 \leq T$ and $ s_1< s_2 <s_1+1$.  
We choose a $\Psi(p,\cdot)$ satisfying \eqref{eq: inf hierarchy} and  $p\geq 1$. 
 We  iterate \eqref{eq: inf hierarchy} $k$ times to obtain
\begin{align*}
 \Psi(p,s_2)- \Psi(p,s_1) &= \sum_{j=1}^k  \frac{(s_2-s_1)^j}{j!} a(j,p) \Psi(p+j,s_1)  \\
 &\quad + \int_0^{s_2-s_1} \dd t_1\cdots\int_0^{t_{k}}\dd t_{k+1} a(k+1,p) \Psi(p+k+1,s_1+t_{k+1}) \,.
\end{align*}
where $a(k,p) = (-1)^k(p+k-1)! / (p-1)!$. By the induction hypothesis, $\Psi(p+j,s_1)= (1+s_1)^{-(p+j)}$ and hence the sum on the right-hand side is the $k$th order Taylor polynomial of the function $s \mapsto (1+s)^{-p}$ at $s=s_1$. The series is absolutely convergent for $s_2-s_1 < 1+s_1$. 

Upon using $\sup_p \norm \Psi(p,\cdot)\norm_\infty \leq 1$, the second term is bounded as
\begin{equation*}
\int_0^{s_2-s_1} \dd t_1\cdots\int_0^{t_{k}}\dd t_{k+1}  \streep a(k+1,p) \streep 
\leq \frac{(s_2-s_1)^{k+1}}{(k+1)!} \frac{(p+k)!}{(p-1)!} 
\end{equation*}
which converges to zero as $k\to\infty$ whenever $s_2-s_1<1$.  Therefore  $ \Psi(p,s_2)=(1+s_2)^{-p}$, completing the induction step.

\end{proof}

%%%%%%%%%%%%%%%%%%%%%%%%%%%%%%%%%%%%%%%%%%%%%%%%%%%%%%%%%%%%

\section{A microscopic model}\label{sec: micro}

%%%%%%%%%%%%%%%%%%%%%%%%%%%%%%%%%%%%%%%%%%%%%%%%%%%%%%%%%%%%

The DMPK theory is a macroscopic theory based on few symmetry assumptions, but does not refer to any particular physically relevant microscopic (Hamiltonian) model. We now introduce a concrete quantum lattice model with disorder, identify the physical symmetries, define the corresponding ensemble of transfer matrices and study its properties for long wires. In the relevant weak coupling limit, we shall derive the stochastic differential equation to be compared with the DMPK evolution. 

%%%%%%%%%

\subsection{The Hamiltonian; symmetries and spectrum}\label{sub: Hamiltonian}

The system is an infinitely extended wire, modeled by the Hilbert space
\begin{equation*}
\caH = l^2 (\bbZ\times\bbZ_N ) = l^2(\bbZ)\otimes \bbC^N\,.
\end{equation*}
A vector $\Psi\in\caH$ is a sequence $\Psi(x,z)$, with longitudinal coordinate $x$ and transverse coordinate $z$, or rather $\Psi_x(z)$ if we prefer to think of a $\bbC^N$-valued sequence. The Hamiltonian
\begin{equation*}
H = H_{\mathrm{kin}} + \lambda V
\end{equation*}
has a deterministic, translation invariant kinetic term
%%
%\begin{equation*}
%(H_{\mathrm{kin}}\Psi)_x = P\str \Psi_{x+1} + H_\perp \Psi_x + P\Psi_{x-1}
%\end{equation*}
%%
and a random on-site potential
\begin{equation*}
(V\Psi)(x,z) = V(x,z)\Psi(x,z)\,.
\end{equation*}
The disorder is limited to a finite region, namely $V(x,z) = 0$ for $x\notin\{1,\ldots,L\}$. The non vanishing elements $V(x,z)$ are i.i.d. real random variables, with  $\bbE(V(x,z)) = 0$ and normalized to have $\bbE(V(x,z)^2) = 1$ so that the strength of the disorder is exclusively controlled by the parameter $\lambda$.

The specific form of the kinetic Hamiltonian will only play a role in determining the symmetry class to which the system belongs, and could therefore be left essentially open, up to these limited symmetry requirements. For simplicity and definiteness, we shall however make here a particular choice that allows for an explicit tracking of the symmetries and their consequences. Henceforth $H_\mathrm{kin}$ will describe a nearest neighbor and diagonal hopping in the presence of a magnetic field in the longitudinal direction, namely
\begin{multline}
\left(H_\mathrm{kin}\Psi\right)(x,z) = \Psi(x+1,z) + \Psi(x-1,z)  \\
+ h_1 \left[ \ep{\iu\gamma}\Psi(x,z+1) + \ep{-\iu\gamma}\Psi(x,z-1) \right] + h_2 \left[ \ep{\iu\gamma}\Psi(x-1,z+1) + \ep{-\iu\gamma}\Psi(x+1,z-1)) \right]  \label{eq: kinetic energy}
\end{multline}
where $h_1,h_2>0$ and $0\leq\gamma < 2\pi / N$. Note that as indicated in the definition of $\hilb$, periodic boundary conditions in the transverse direction are imposed on all operators. \\

As we have briefly discussed in Section~\ref{sec: DMPK}, the DMPK equation comes in various guises, depending on the abstract symmetry group of the transmission matrix. At the microscopic level, time reversal $T$ is naturally implemented by complex conjugation in the `position basis', namely
\begin{equation*}
(T\Psi)(x,z) = \overline{\Psi(x,z)}\,.
\end{equation*}
It is immediate to check that $H_\mathrm{kin}$ is invariant under $T$ iff $\gamma  =0$: the magnetic field indeed breaks time reversal invariance.

\subsubsection{Eigenvalues, eigenvectors and chaoticity}

We consider the eigenvalue problem for the kinetic Hamiltonian
\begin{equation}
 \label{evp}
H_{\mathrm{kin}}\Psi=E\Psi
\end{equation}
at some fixed energy $E$. 
 By translation invariance and the periodic boundary conditions, the (non-normalizable) solutions are given by
\begin{equation*}
\Psi_{k,\nu}(x,z)=\frac{1}{\sqrt{N}}e^{ik x}e^{\frac{2\pi i}{N}\nu z}
\end{equation*}
for $\nu=1,\ldots,N$, with
\begin{equation}
\label{momenta} 
E =E(k,\nu)= 2\cos(k) + 2h_1\cos\left(\gamma + \frac{2\pi}{N}\nu\right) + 2 h_2\cos\left(k - \gamma - \frac{2\pi}{N}\nu\right)
\end{equation}
We now look for the condition on $E$ so that \eqref{evp} has  plane wave solutions, rather than exponentially decaying ones, i.e.\ such that $E=E(k,\nu)$ for some $k,\nu$. In physical terms, this means that we do not want to study evanescent modes, also called `elliptic channels'.

In fact, we first fix the energy $E$ and consider solutions $k = k_\nu(E)$ of~(\ref{momenta}) for any `transversal mode' $\nu$. We shall drop the $E$ dependence in the sequel. There are two such wavevectors $k_\nu^\pm$ corresponding to a right moving and a left moving wave. The relations
\begin{equation} \label{TRIrelations}
k_\nu^+ = -k_{-\nu}^-
\end{equation}
holds in the time reversal invariant case, but is broken if $\gamma>0$. If $h_2=0$, the residual symmetry $\Psi(x,z) \mapsto \Psi(x,-z)$ induces the additional degeneracy $k_\nu^+ = -k_\nu^-$ that needs to be avoided. It is a straightforward exercise to check
\begin{lma}
 \label{lmawavevectors}
For an energy $E\neq0$, $|E|<2$, and a kinetic Hamiltonian $H_{\mathrm{kin}}$ with parameters $0\leq\gamma<\frac{\pi}{N}$, and $h_1,h_2>0$ sufficiently small, in particular
\begin{equation}\label{energyrestriction}
|E|+2h_1+2h_2<2\,,
\end{equation}
the equation
\begin{equation*}
H_{\mathrm{kin}}\Psi=E\Psi
\end{equation*}
has $2N$ plane wave solutions
\begin{equation*}
 \Psi_{k,\nu}(x,z)=\frac{1}{\sqrt{N}}e^{ik_\nu^{\sigma} x}e^{\frac{2\pi i}{N}\nu z}
\end{equation*}
with $\nu\in\bbZ_N$, $\sigma\in\lbrace{+,-\rbrace}$. For $\ga \neq 0$, the longitudinal wave numbers $k_\nu^{\sigma}$ are non degenerate in the sense that
\begin{equation*}
\vs_1 k_{\nu(1)}^{\sigma_1}+\vs_2 k_{\nu(2)}^{\sigma_2}+\vs_3 k_{\nu(3)}^{\sigma_3}+\vs_4 k_{\nu(4)}^{\sigma_4}=0\mbox{ }(\mathrm{mod}\mbox{ }2\pi)
\end{equation*}
for signs $\vs_1,...,\vs_4$ only in the trivial case
\begin{equation}
\label{exceptionTRI}
\big(\vs_1, \sigma_1, \nu(1)\big) = \big(-\vs_2, \sigma_2, \nu(2)\big)\quad\text{and}\quad
\big(\vs_3, \sigma_3, \nu(3)\big) = \big(-\vs_4, \sigma_4, \nu(4)\big)\,,
\end{equation}
and if $\gamma=0$ also in the $T$-symmetric situation
\begin{equation}
\label{exception}
\big(\vs_1, \sigma_1, \nu(1)\big) = \big(\vs_2, -\sigma_2, -\nu(2)\big)\quad\text{and}\quad
\big(\vs_3, \sigma_3, \nu(3)\big) = \big(\vs_4, -\sigma_4, -\nu(4)\big)\,,
\end{equation}
all of this up to relabeling of the indices.
\end{lma}
The above lemma shows that the Hamiltonian has no \emph{other} symmetry than time reversal invariance, in accordance with the macroscopic theory of the previous section. It is important to realize that a residual symmetry of the kinetic Hamiltonian could, and in fact would, leave a trace in the scaling limit, \emph{even if that symmetry got broken by the disorder}. The DMPK theory considers a single ensemble of transfer matrices, whereas there are really two microscopic models, one without and one with disorder. For our derivation, it is essential that both have the correct symmetry properties.
 
For further considerations, it is useful to define the chaoticity of the kinetic Hamiltonian in analogy to~\cite{viragvalkoexplosion} by 
\begin{equation}
 \label{defcha}
\cha\left(\gamma,h_1,h_2\right)=\min\left\lbrace\left|\sum_{i=1}^4 \vs_i k_{\nu(i)}^{\sigma_i}\right|, ({\vs},{\sigma},{\nu})\mbox{ not solving } (\ref{exceptionTRI}) \right\rbrace\,,
\end{equation}
for $\gamma>0$ and
\begin{equation}
 \label{defchawithout}
\cha\left(0,h_1,h_2\right)=\min\left\lbrace\left|\sum_{i=1}^4 \vs_i k_{\nu(i)}^{\sigma_i}\right|,
({\vs},{\sigma},{\nu})\mbox{ not solving } (\ref{exceptionTRI})\mbox{ or } (\ref{exception})\right\rbrace.
\end{equation}
if $\gamma=0$.

%%%%%%%%%

\subsection{The transfer matrix}\label{sub: MicroTransfer}

We decompose the kinetic Hamiltonian
\beq
(H_\mathrm{kin}\Psi)_x=  H_\perp  \Psi_x +  P \Psi_{x+1} + P^* \Psi_{x-1} 
\eeq
as a strictly transverse operator
\begin{equation*}
(H_\perp\phi)(z) = h_1 \left( \ep{\iu\gamma}\phi(z+1) + \ep{-\iu\gamma}\phi(z-1) \right)\,,
\end{equation*}
and the components inducing hopping in the longitudinal direction
\begin{equation*}
(P\phi)(z) = \phi(z) + h_2 \ep{\iu\gamma} \phi(z+1)\,,\quad\text{and}\quad(P\str\phi)(z) = \phi(z) + h_2 \ep{-\iu\gamma} \phi(z-1)\,,
\end{equation*}
where $\phi(z)\in\bbC^N$. Similarly, the random potential can be seen as a sequence of $N\times N$ diagonal matrices $V_x$. With these notations, the eigenvalue equation with disorder, $\lambda>0$, reads
\begin{equation*}
\begin{pmatrix}\Psi_{x+1}\\\Psi_{x}\end{pmatrix}=
\begin{pmatrix}\left(P^*\right)^{-1}(E-H_\perp-\lambda V_x)&-\left(P^*\right)^{-1}P\\{1}&{0}\end{pmatrix}\begin{pmatrix}\Psi_{x}\\\Psi_{x-1}\end{pmatrix}=:T_x^\lambda\begin{pmatrix}\Psi_{x}\\\Psi_{x-1}\end{pmatrix},
\end{equation*}
thereby defining the transfer matrix $T_x^\lambda$ (of dimension $2N$) for the layer $x$ as it usually appears in the mathematical literature. 
By the multiplicative property, the transfer matrix for $L$ layers is simply given by the product of the one-layer matrices,
\begin{equation*}
 \begin{pmatrix}\Psi_{L+1}\\\Psi_{L}\end{pmatrix}=T_L^\lambda\cdots T_1^\lambda\begin{pmatrix}\Psi_{1}\\\Psi_{0}\end{pmatrix},
\end{equation*}
In order to allow for a comparison with the DMPK theory, it is important to write the transfer matrix in the natural basis for $\bbC^{2N}$, i.e. the basis in which the first $N$ components correspond to left moving waves, and the other $N$ components are right moving. Morever, in the time-reversal invariant case we need the correct identification between left and right moving channels, given by~(\ref{TRIrelations}). In other words, the time reversal of the vector 
$
(\phi_1, \phi_2)^t
$ 
must correspond to
$
(\bar{\phi}_2, \bar{\phi}_1)^t
$.
Only in that basis do the definitions of the transfer matrix given in Section \ref{sub: Transfer} apply. We shall refer to it as the \emph{channel basis}.

Let $Q$ be the $N\times N$ matrix whose columns are the transverse eigenvectors $\phi_\nu(z) = 1/\sqrt{N}\exp(2\pi \iu\nu z/N)$. Furthermore, let $\Pi$ be the permutation matrix that interchanges the $\nu$ and $-\nu$ channels, and let
\begin{equation*}
\Upsilon :=
 \begin{pmatrix}\frac{\ep{\iu k^+}}{\sqrt{\left|v^+\right|}}&\frac{\ep{\iu k^-}}{\sqrt{\left|v^-\right|}}\\\frac1{\sqrt{\left|v^+\right|}}&\frac1{\sqrt{\left|v^-\right|}}\end{pmatrix}
 \begin{pmatrix}1&0\\0&\Pi\end{pmatrix} 
\end{equation*}
where $k^\pm$ and $v^\pm$ denote diagonal matrices with elements $k^\pm_{\nu}$ and $v^\pm_\nu$, and $v^\pm_\nu$ are the velocities in each channel,
\begin{equation*}
v^\pm_\nu=\left.\frac{\partial}{\partial k}E(k,\nu)\right|_{k=k_\nu^\pm}=-2\sin\left(k^\pm_\nu\right)-2h_2\sin\left(k^\pm_\nu - \frac{2\pi}{N}\nu - \gamma\right).
\end{equation*}
Note that all $v^\pm_\nu$ are nonzero by (\ref{energyrestriction}). The action of $T$ (complex conjugation) on the eigenvectors yields
\begin{equation*}
T:\,(k_\nu^\sigma,\nu)\mapsto(-k_\nu^\sigma,-\nu)
\end{equation*}
which equals $(k_{-\nu}^{-\sigma},-\nu)$ if time reversal invariance holds, so that the $(\sigma = -)$ block needs to be reordered by $\nu\leftrightarrow-\nu$. Hence the permutation $\Pi$. With these notations, Lemma~\ref{lmawavevectors} reads
\begin{equation} \label{diagsecond}
\Upsilon^{-1}(Q\otimes 1)\str T_x^0(Q\otimes 1)\Upsilon = \begin{pmatrix}\ep{\iu k^+}&0\\0&\Pi \ep{\iu k^-} \Pi\end{pmatrix} =: M_x^0\,.
\end{equation}
The matrix $M_x^0$ is the transfer matrix for the purely kinetic transport in a single layer, which is diagonal the channel basis. As expected, it simply adds a phase to the traveling plane wave.

It is now easily checked that $M_x^0$ is a bonafide transfer matrix in the sense of Definition~\ref{Def:Transfer}: Eq.~(\ref{PseudoUnitary}) is always satisfied and (\ref{M_TRI}) holds at $\gamma=0$, as conjugation by $\Sigma_x$ precisely maps
$k^\sigma_\nu$ to $k^{-\sigma}_{-\nu}=-k^\sigma_{\nu}$. 
 The transfer matrix for the total system is again obtained by multiplication
\begin{equation*}
M^\lambda(L) :=  M^\lambda_L \cdots  M^\lambda_1   =   K^{-1} \, T_L^\lambda\cdots T_1^\lambda \, K\,,
\end{equation*}
where $K = (Q\otimes 1)\Upsilon$ is the change of basis from the position basis to the channel basis.
This random matrix and its relation the DMPK theory is the central object of study in the following.
%%%%%%%%%

\subsection{A scaling limit}\label{sub: scaling}

The DMPK theory suggests that the microscopic transfer matrix for a disordered wire of length $L$ should converge to a solution of the DMPK equation in the correct macroscopic limit. As discussed in the introduction, the natural scaling between the microscopic length $L$ and the macroscopic length $s$ is through the mean free path, $L = \lambda^{-2}s$. A naive interpretation of the DMPK theory would then be the convergence of $M^\lambda(\lfloor \lambda^{-2}s \rfloor)$ to $\caM(s)$. This cannot possibly hold as $M^\lambda(\lfloor \lambda^{-2}s \rfloor)$ contains rapidly oscillating terms as a function of $s$, as already exemplified at the level of the unperturbed system (\ref{diagsecond}):
\begin{equation*}
 M^0(\lfloor \lambda^{-2}s \rfloor)=\begin{pmatrix}\exp\left(\iu \lfloor \lambda^{-2}s \rfloor k^+\right)&0\\0&\Pi \exp\left(\iu \lfloor \lambda^{-2}s \rfloor k^-\right) \Pi\end{pmatrix}\,.
\end{equation*}
To obtain a reasonable limit, we therefore consider
\begin{equation*}
 A^\lambda(\lfloor \lambda^{-2}s \rfloor) := (M^0(\lfloor \lambda^{-2}s \rfloor))^{-1}M^\lambda(\lfloor \lambda^{-2}s \rfloor).
\end{equation*}
As the set of matrices of Definition~\ref{Def:Transfer} form a group, $A^\lambda$ is a transfer matrix again. Moreover, it is an easy check that in the polar decomposition \eqref{MSdecomp}, the matrix $S$ corresponding to $A^\lambda$ is the same as that corresponding to  $M^\lambda$,  so that they have the same transmission eigenvalues.

In order to state the result of the scaling limit, $\lambda\to0$, we introduce the following processes. For $\gamma>0$, let
\begin{equation*}
 \caZ_{\gamma}(s) := \begin{pmatrix} \mathfrak{a}(s) & \mathfrak{b}(s) \\\mathfrak{b}^*(s) &{ \mathfrak{a}'}(s)\end{pmatrix}
\end{equation*}
with 
\begin{equation*}
 \mfa_{\mu\mu}(s)=-\mfa'_{\mu\mu}(s)=\frac{\iu}{\sqrt{(4-E^2)N}}W(s)
\end{equation*}
with the \emph{same} standard real Brownian motion $W$ for all $\mu=1,\ldots,N$. All diagonal elements of $\caZ$ are thus perfectly correlated. For the off-diagonal elements, 
\begin{align*}
 \mfa_{\mu\nu}(s)=-\overline{\mfa_{\nu\mu}(s)}&=\frac{1}{\sqrt{(4-E^2)N}}B^{++}_{\mu\nu}(s)&\mbox{ for }&1\leq\mu<\nu\leq N\\
\mfa'_{\mu\nu}(s)=-\overline{\mfa'_{\nu\mu}(s)}&=\frac{1}{\sqrt{(4-E^2)N}}B^{--}_{\mu\nu}(s)&\mbox{ for }&1\leq\mu<\nu\leq N\\
\mfb_{\mu\nu}(s)&=\frac{1}{\sqrt{(4-E^2)N}}B^{+-}_{\mu\nu}(s)&\mbox{ for }&1\leq\mu,\nu\leq N
\end{align*}
 with all the elements of $B^{++},B^{--}$ and $B^{+-}$ standard complex Brownian motions, mutually independent and independent of $W$. For $\gamma=0$,
\begin{equation*}
 \caZ_0(s) := \begin{pmatrix} \mathfrak{a}(s) & \mathfrak{b}(s) \\ \overline{\mathfrak{b}(s)} &\overline{ \mathfrak{a}(s)}\end{pmatrix}
\end{equation*}
with the definition of $\mfa$ unchanged, but $\mfa'=\overline{\mfa}$ now, and a symmetric $\mfb$:
\begin{align*}
\mfb_{\mu\nu}(s)=\mfb_{\nu\mu}&=\frac{1}{\sqrt{(4-E^2)N}}\tilde{B}^{+-}_{\mu\nu}(s)&\mbox{ for }&1\leq\mu\leq\nu\leq N,
\end{align*}
with the entries of the $\tilde{B}^{+-}$ independent standard complex Brownian motions again, independent of the elements of $B^{++}$ and $W$.

The result of the scaling limit, $\lambda\to0$ is summarized in the main theorem.
\begin{thm}
\label{thm:convA}
If $h_1$ and $h_2$ depend on $\lambda$ so that
\begin{equation}
\label{tozero}
h_1(\lambda) \longrightarrow 0\,,\qquad
h_2(\lambda) \longrightarrow 0, \qquad\text{and}\qquad \lambda^{-2}\cha(\gamma,h_1(\lambda),h_2(\lambda))\longrightarrow\infty\,,
\end{equation}
as $\lambda\rightarrow0$, 
then the process $\left(A^\lambda\left(\left\lfloor\lambda^{-2}s\right\rfloor\right)\right)_{s\geq0}$ converges in distribution to the process $\left(\caA(s)\right)_{s\geq0}$
%%
%\begin{equation*}
% A^\lambda\left(\lambda^{-2}\cdot\right)\stackrel{\caL}{\longrightarrow}\caA
%\end{equation*}
%
on the path space of $\bbC^{2N\times 2N}$-valued processes endowed with Skorhod topology.

\noindent{}For $\gamma\geq0$, $\left(\caA(s)\right)_{s\geq0}$ is given as the unique solution for $s\geq0$ to
\begin{equation}
 \label{defcaA}
\begin{split}
\dd\caA(s)&=\dd\caZ_\gamma(s)\caA(s)\\
\caA(0)&=1.
\end{split}
\end{equation}
\end{thm}
\remark Existence and uniqueness of $\caA$ is a standard result, as all entries of $\caZ$ are Brownian motions. Note that no additional moment condition is needed on the random variables $V(x,z)$. 

The convergence of the hopping parameters $h_i$ to zero merely brings the resulting process $\caA$ into a isotropic form close to the DMPK process. It is however not essential for the scaling limit per se, and the same techniques used in the proof of Theorem~\ref{thm:convA} yield a limiting process for fixed $h_1, h_2\neq 0$, but a less isotropic one. For $\gamma>0$, 
\begin{equation*}
 \caY_{\gamma}(s) := \begin{pmatrix} \alpha(s) & \beta(s) \\ \beta(s)\str &{ \alpha'}(s)\end{pmatrix}
\end{equation*}
where $\alpha(s), \alpha'(s)$ and $\beta(s)$ differ from $\mfa(s),\mfa'(s)$ and $\mfb(s)$ only through their covariances, namely by the replacement
\begin{equation*}
\frac{1}{\sqrt{(4-E^2)N}} \quad \longrightarrow \quad \frac{1}{\sqrt{N \left|v^+_{\sigma_1\nu_1}v^+_{\sigma_2\nu_2}\right|}}\,.
\end{equation*}
For $\gamma=0$,
\begin{equation*}
 \caY_0(s) := \begin{pmatrix} \alpha(s) & \beta(s) \\[1mm] \overline{\beta(s)} &\overline{\alpha(s)}\end{pmatrix}
\end{equation*}
with the same substitution.
\begin{prop} \label{Prop:scalingNIsotropic}
Let $(\gamma, h_1, h_2)$ be fixed (in particular, not dependent on $\la$) and such that Lemma~\ref{lmawavevectors} holds. As $\la \to 0$,  the process $\left(A_\lambda\left(\left\lfloor\lambda^{-2}s\right\rfloor\right)\right)_{s\geq0}$ converges in distribution to $\left(\caG(s)\right)_{s\geq0}$ on the path space of $\bbC^{2N\times 2N}$-valued processes endowed with Skorhod topology. $\left(\caG(s)\right)_{s\geq0}$ is given as the unique solution for $s\geq0$ to
\begin{equation}
 \label{defcaB}
\begin{split}
\dd\caG(s)&=\dd\caY(s)\caG(s) \\
\caG(0)&=1. \\
\end{split}
\end{equation}
\end{prop}

\subsection{Discussion} \label{sec: discussion}

Let us comment on the limiting process of transfer matrices $\caA(s)$, and compare it to the ideal ensemble $\caM$ of Section~\ref{sec: DMPK}. As already discussed in~\cite{BdR}, the overall factor $\sqrt{4-E^2}$ only corresponds to a redefinition of the mean free path and is irrelevant here. The major difference lies in the diagonal of the processes $\mathfrak{a}, \mathfrak{a}'$ generating $\caA$, which have \emph{perfectly correlated} diagonal elements, whereas they are independent of each other in their cousins generating $\caM$. In the case $\be=1$, an additional deviation can be found in the variance of the diagonal elements of $\mathfrak{b}$, which are smaller here than in the ideal case by a factor $\sqrt{2}\cdot\sqrt{N/(N+1)}$. Despite these differences, we have:
\begin{cor} If $\be=2$, the law of the process of transmission eigenvalues $(T_k)_{k=1}^N$ induced by $\caA$ is the same as that induced by $\caM$. 
\end{cor}
Mathematically, this can be observed in the fact that the SDE for $\caM_{++}\str\caM_{++}$, obtained in a straightforward way by It\^o calculus,
\begin{equation*}
\dd (\caM_{++}\str\caM_{++} )= \caM_{++}\str \dd \mathfrak{b}\caM_{-+} + \caM_{-+}\str \dd \mathfrak{b}\str\caM_{++} + \caM_{++}\str\caM_{++} \dd s + \caM_{-+}\str\caM_{-+} \dd s
\end{equation*}
does not depend on $\dd \mathfrak{a}^\sharp$, and by recalling that the transmission eigenvalues are directly related to the eigenvalues $\lambda_k$ of this matrix. For $\beta=2$, the limiting $\dd \mathfrak{b}$ is exactly equal to the ideal one. The underlying physical reason is that the $\mathfrak{a}^\sharp$ blocks in the infinitesimal transfer matrix merely change the basis of left-, respectively right-moving channels. As such they do not change the magnitude of either the scattered or the reflected waves, hence they do not contribute to the transmission eigenvalues $T_{k}$. 

For the $\be=1$ case, the process of transmission eigenvalues induced by $\caA$ is different from that induced by $\caM$. 
Of course one may suspect that in the limit $N \to \infty$, these processes have similar features, in particular, that the variance of the conductance in the diffusive regime is the same in both cases, and equal to its universal value. In fact in the same limit, we even expect the non-isotropic processes of Proposition~\ref{Prop:scalingNIsotropic} to share the same property. This remains currently at the level of speculations, the reason being the relatively poor understanding of the properties of the DMPK equation itself in the large $N$ limit. This should be contrasted with recent efforts in the study of hermitian random matrices, where the Gaussian ensembles are very well-known and the challenge is to show that other ensembles share some of their properties, see e.g.\ \cite{ErYa12} and references therein.

\subsubsection{Symmetry considerations}

As already mentioned in Section \ref{sec: DMPK}, the  driving process $\caL$ in the SDE for $\caM$ satisfies the invariance property
 \beq  \label{eq: invariance}
 \caW  \caL \caW^{-1} \mathop{=}\limits^{d}  \caL
 \eeq 
 for 
 \beq
 \caW =  \begin{pmatrix} \caW_+ &0 \\  0 &\caW_-  \end{pmatrix}
 \eeq
 with $\caW_+, \caW_-$ arbitrary unitaries in the $\beta=2$ case and satisfying $\caW_-=\overline{\caW_+}$ in the $\beta=1$ case. In other words, all channel bases are assumed to be equivalent. It is exactly this equivalence that is lost in our model. This can be understood heuristically as follows. Here, the size of the impurities is assumed to be much smaller than the wavelength of the scattered waves, a fact expressed by the $\delta$-correlation in space of the potential, i.e.\ $V=\sum_y V_y$ with $y=(x,z)$ and $\bbE(V_yV_{y'})= \delta_{y,y'}$. Therefore, the position basis (or its dual, the momentum basis) is naturally singled out. There is no reason to expect another choice of basis to be equivalent and to allow for arbitrary unitaries $\caW_+,\caW_-$ in~\eqref{eq: invariance}. The symmetry that still survives (but in fact, only so because we performed the additional scaling limit $h_1, h_2 \to 0$) is the relabeling of channels; indeed the scaling limit $h_1, h_2 \to 0$ makes the group velocity of all channels $\nu$ equal at a given energy $E$. 
Up to phases, this corresponds to restricting  $\caW_+,\caW_-$ to be permutation matrices. One easily checks that, with this restriction, the driving process $\caZ_\ga$ still satisfies \eqref{eq: invariance}.

Another consequence of the locality of the impurities is captured by the following heuristic argument. Let us consider the kernel $S_{p,p'}$ of the scattering matrix in momentum basis, i.e.\ $p=(k,\nu)$.
  Since the scattering is weak (as already indicated, this is inherent to the setup and it is forced in our model by the $\la\to 0$ limit), multiple scattering can be neglected for short slabs of material, hence we can use the  Born approximation for $S_{p,p'}$:
   \beq
     S_{p,p'} =  \delta_{p,p'}-2 \pi \iu  \la   \sum_y  (\tilde V_y)_{p,p'} \delta(E(p)-E(p'))  + \caO(\la^2)
   \eeq
  where $(\tilde V_y)=V_ye^{\iu (p-p')y}$.   For $p=p'$, this expression depends on $p$ only through the factor $\delta(E(p)-E(p'))$ which contributes a $p$-dependent group velocity, a dependence which vanishes in the subsequent scaling limit $h_1,h_2 \to 0$. This means that the diagonal elements of the transmission matrix all coincide and this is exactly what we find in the ensemble $\caA(s)$, since in lowest order the $\mathfrak{a}$ block corresponds  to the transmission matrix.
  
Finally, we mention~\cite{RoSc10,SaSc10} where geometric methods are used to analyze a similar weak coupling limit of the Anderson model on tubes, and contact is made with random matrix theory. 
  
%%%%%%%%%%%%%%%%%%%%%%%%%%%%%%%%%%%%%%%%%%%%%%%%%%%%%%%%%%%%

\section{Proof of transfer matrix limits} \label{sec: proofs}
We prove Theorem \ref{thm:convA}.

The evolution of $\left(A^\lambda(x)\right)_{x\geq0}$ is given by $A^\lambda(0)=1$ and the stochastic difference equation
\begin{align}
\label{evolA}
A^\lambda(x)-A^\lambda(x-1) &= \left((M^0(x))^{-1} K^{-1}T^\lambda_x K M^0(x-1)-1\right) A^\lambda(x-1) \nonumber \\
&=: \lambda Z_x A^\lambda(x-1),
\end{align}
where $Z_x = (M^0(x))^{-1} R_x M^0(x)$ and we defined
\begin{equation*}
\lambda R_x := K^{-1}T^\lambda_x K (M^0_x)^{-1} -1 = K^{-1}(T^\lambda_x-T^0_x) K (M^0_x)^{-1}.
\end{equation*}
Recall that $K$, introduced in the previous section, stands for the change from the position basis to the channel basis. It follows that 
\begin{equation*}
R_x = \Upsilon^{-1} \begin{pmatrix}-Q\str(P\str)^{-1}V_xQ & 0 \\ 0 & 0\end{pmatrix}\Upsilon(M^0_x)^{-1}
\end{equation*}
Using the explicit forms of $\Upsilon, Q$ and $P$, this matrix reads
\begin{equation}
\label{Rxsimplified}
 R_x=\iu \begin{pmatrix}\frac{1}{\sqrt{\left|v^+\right|}}Q^*V_xQ\frac{1}{\sqrt{\left|v^+\right|}}&\frac{1}{\sqrt{\left|v^+\right|}}Q^*V_xQ\frac{1}{\sqrt{\left|v^+\right|}}\Pi\\-\Pi\frac{1}{\sqrt{\left|v^+\right|}}Q^*V_xQ\frac{1}{\sqrt{\left|v^+\right|}}&-\Pi\frac{1}{\sqrt{\left|v^+\right|}}Q^*V_xQ\frac{1}{\sqrt{\left|v^+\right|}}\Pi\end{pmatrix},
\end{equation}
which of course satisfies $(R_{++})\str = -R_{++}$, $(R_{--})\str = -R_{--}$ and $(R_{-+})\str = R_{+-}$. \\
In the time reversal invariant case, since
\begin{equation*}
 Q=Q^T,\qquad Q\Pi=Q^*,\qquad\Pi Q^*=Q
\end{equation*}
this simplifies to
\begin{equation}
\label{Rxsimplifiedgamma0}
 R_x=\iu\begin{pmatrix}\frac{1}{\sqrt{\left|v^+\right|}}Q^*V_xQ\frac{1}{\sqrt{\left|v^+\right|}}&\frac{1}{\sqrt{\left|v^+\right|}}Q^*V_xQ^*\frac{1}{\sqrt{\left|v^+\right|}}\\-\frac{1}{\sqrt{\left|v^+\right|}}QV_xQ\frac{1}{\sqrt{\left|v^+\right|}}&-\frac{1}{\sqrt{\left|v^+\right|}}QV_xQ^*\frac{1}{\sqrt{\left|v^+\right|}}\end{pmatrix}\,.
\end{equation}
Since $Q^*=\overline{Q}$, we also have that $R_{--} =\overline{ R_{++}}$, and in a similar fashion $R_{-+} = \overline{R_{+-}}$.

Before we go further into the proof, let us explain the heuristics of the convergence to the DMPK equation. The matrix $R_x$ contains the $N$ i.i.d random variables $V(x,z)$, $z=1,\ldots,N$, and $R_x,R_y$ are independent for $x\neq y$. Under the appropriate technical conditions, we have convergence, as  $ \lambda\to0$;
\begin{equation*}
\lambda\sum_{x=0}^{\lfloor\lambda^{-2}s\rfloor}R_x \stackrel{\mathrm{d}}{\longrightarrow}\caR(s)\,.
\end{equation*}
where the $V(x,z)$ in $R_x$ are replaced by Brownian motions $B_z(s)$ in $\caR(s)$. Of course, this is nothing else than the convergence of a random walk to Brownian motion. Now, the matrices $Z_x$ generating the discrete process $A(x)$ do contain the highly oscillating phases of $M^0(x)$. Let
\begin{equation*}
Z^\lambda(s) = \lambda\sum_{x=0}^{\lfloor\lambda^{-2}s\rfloor}Z_x\,.
\end{equation*}
The correlation of any two matrix elements reads
\begin{multline}
\label{EZZ}
\bbE\left[\left(Z^\lambda(s)\right)_{mn}\left(Z^\lambda(s)\right)_{pr}\right] \\
=\lambda^{2}\sum_{x=1}^{\lfloor\lambda^{-2}s\rfloor}\exp\left(\iu x\left(-k^{\sigma_m}_{\sigma_m\nu_m}+k^{\sigma_n}_{\sigma_n\nu_n}-k^{\sigma_p}_{\sigma_p\nu_p}+k^{\sigma_r}_{\sigma_r\nu_r}\right)\right)\bbE\left[(R_x)_{mn}(R_x)_{pr}\right].
\end{multline}
where the $\sigma$'s and $\nu$'s denote the block and position in the block of a certain element, the `physical' momentum, due to the permutation $\Pi$ is not $\nu$, but $\sigma\nu$. The expectation in the r.h.s is independent of $x$, so that this sum is highly oscillatory and formally converges to a $\delta$ function. 
Hence, the phases create a limiting process $\caZ(s)$ with almost completely uncorrelated entries, apart from the exceptional conditions of Lemma~\ref{lmawavevectors}. As a result, the number of independent random variables in $\caZ(s)$ is of order $N^2$, whereas it was only $N$ in $\caR(s)$, a phenomenon called `noise explosion' in~\cite{viragvalkoexplosion}. Precisely, we prove:
\begin{lma}
 \label{lma_convergeZ}
Under the conditions of Theorem~\ref{thm:convA} and in the same topology, the process $\left(Z^\lambda(s)\right)_{s\geq0}$ converges in distribution to the process $\left(\caZ_\gamma(s)\right)_{s\geq0}$, for $\gamma\geq0$.
\end{lma}
\begin{proof}
We first recall that $Q$ is the matrix of transversal plane waves so that
\begin{equation*}
(Q\str V_x Q)_{\mu\nu} = \widehat{V_x}(\mu-\nu)\,.
\end{equation*}
In particular, all elements on the diagonal are perfectly correlated.

Now, we consider~(\ref{EZZ}) for $s\geq0$, and note that the wavevectors $k_\nu^\sigma$ and the matrices $R_x$ all depend on $\lambda$ implicitly through the dependence $h_i(\lambda)$. However,  $\bbE\left[(R_x)_{mn}(R_x)_{pr}\right]$ is independent of $x$ and can be taken out of the sum. For all choices $\{(\nu_i,\sigma_i): i=m,n,p,r\}$ for which the exponent does not vanish, we have that
\begin{equation*}
\left| \lambda^{2}\sum_{x=1}^{\left\lfloor\lambda^{-2}s\right\rfloor}\exp\left(ix\left(-k^{\sigma_m}_{\sigma_m\nu_m}+k^{\sigma_n}_{\sigma_n\nu_n}-k^{\sigma_p}_{\sigma_p\nu_p}+k^{\sigma_r}_{\sigma_r\nu_r}\right)\right)\right|\leq\frac{2\lambda^2}{\cha(\lambda)}\longrightarrow0,
\end{equation*}
as $\lambda\rightarrow0$, by Assumption~(\ref{tozero}). By Lemma \ref{lmawavevectors}, the the wavenumbers cancel out whenever
\begin{equation}
\label{sigmanu1}
 \left(\sigma_m,\nu_m\right)=\left(\sigma_n,\nu_n\right)\mbox{ and }\left(\sigma_p,\nu_p\right)=\left(\sigma_r,\nu_r\right)\,,
\end{equation}
corresponding to any two diagonal elements in $Z^\lambda$, or
\begin{equation}
\label{sigmanu2}
 \left(\sigma_m,\nu_m\right)=\left(\sigma_r,\nu_r\right)\mbox{ and }\left(\sigma_p,\nu_p\right)=\left(\sigma_n,\nu_n\right)\,,
\end{equation}
namely two mutually `transpose' entries of $Z^\lambda$. In the case $\gamma=0$, a last possibility is given by
\begin{equation}
\label{sigmanu3}
 \left(\sigma_m,\nu_m\right)=\left(-\sigma_p,\nu_p\right)\mbox{ and }\left(\sigma_n,\nu_n\right)=\left(-\sigma_r,\nu_r\right)\,,
\end{equation}
corresponding to two elements in the same position, but within the opposite blocks. In these three cases and for all $\lambda$,
\begin{equation*}
 \lambda^{2}\sum_{x=1}^{\left\lfloor\lambda^{-2}s\right\rfloor}1=s.
\end{equation*}

Furthermore, using the explicit form~(\ref{Rxsimplified}) of $R_x$ and the definition of $Q$, 
\begin{multline*}
\bbE\left[(R_x)_{mn}(R_x)_{pr}\right]
=-\frac{\sigma_m\sigma_p}{N^2\sqrt{\left|v^+_{\sigma_m\nu_m}v^+_{\sigma_n\nu_n}v^+_{\sigma_p\nu_p}v^+_{\sigma_r\nu_r}\right|}} \\ 
\cdot\sum_{z=0}^{N-1}\exp\left(\frac{2\pi i}{N}z\left(-\sigma_m\nu_m+\sigma_n\nu_n-\sigma_p\nu_p+\sigma_r\nu_r\right)\right)
\end{multline*}
which we only need to evaluate in the `stationary phase' situations (\ref{sigmanu1}), (\ref{sigmanu2}), and (\ref{sigmanu3}). In all these cases,
\begin{equation}
\label{fourvs}
\bbE\left[(R_x)_{mn}(R_x)_{pr}\right] = -\frac{\sigma_m\sigma_p}{N\sqrt{\left|v^+_{\sigma_m\nu_m}v^+_{\sigma_n\nu_n}v^+_{\sigma_p\nu_p}v^+_{\sigma_r\nu_r}\right|}}.
\end{equation}
Note that the $v^+_{\nu}$ still depend on $\lambda$ and $\nu$, but by (\ref{tozero}), they converge to the $\nu$-independent limit $|2\sin(k)|$ with $2\cos(k)=E$. Hence, in the limit $\lambda\rightarrow0$, we have
\begin{equation*}
\bbE\left[(R_x)_{mn}(R_x)_{pr}\right] \longrightarrow -\frac{\sigma_m\sigma_p}{(4-E^2)N}\,.
\end{equation*}
A very similar oscillatory sum appears in $\bbE\left[\left(Z^\lambda(s)\str\right)_{mn}\left(Z^\lambda(s)\right)_{pr}\right]$, with parallel conclusions. In summary, for all $s\geq0$
\begin{equation}
\label{convbracket}
 \begin{split}
  \lim_{\lambda\rightarrow0}\bbE\left[\left(Z^\lambda(s)\right)_{mn}\left(Z^\lambda(s)\right)_{pr}\right]
  &=\int_0^s\dd\left\langle \left(\caZ\right)_{mn}, \left(\caZ\right)_{pr}\right\rangle_t,\\
\lim_{\lambda\rightarrow0}\bbE\left[\left(Z^\lambda(s)\str\right)_{mn}, \left(Z^\lambda(s)\right)_{pr}\right]
&=\int_0^s\dd\left\langle \left(\caZ^*\right)_{mn},\left(\caZ\right)_{pr}\right\rangle_t.
 \end{split}
\end{equation}
with $\caZ$ given in the previous section, and $\langle M,N \rangle_t$ the bracket process of two martingales $M,N$.  In particular, the block structure of $\caZ$ arises from the corresponding relations noted above in $R_x$. Moreover, perfect correlation of the diagonal elements of $Q\str V_x Q$ and the exceptional case~(\ref{sigmanu1}) imply the perfect correlation of the diagonal elements of $\dd \mfa$ and $\dd \mfa'$. All other exceptional cases impose the correlations $\vert \dd \mfb_{\mu\nu}\vert^2$ and $\vert \dd \mfa_{\mu\nu}\vert^2$.

Now the lemma follows as a simple generalization of Donsker's invariance principle, for example by using Chapter VII, Theorem 3.7 in~\cite{js}. They check the convergence of three characteristics, of which, in their notation, $\left[\sup-\beta_3'\right]$ is trivially fulfilled as $Z^\lambda$ and $\caZ$ are martingales, $\left[\gamma_3'-\bbR^+\right]$ is the convergence of brackets as shown in~(\ref{convbracket}), and $\left[\delta_{3,1}-\bbR^+\right]$ is a simple estimate on the jumps of $Z^\lambda$, namely
\begin{equation*}
 \lim_{\lambda\rightarrow0}\sum_{x=1}^{\lfloor\lambda^{-2}s\rfloor}\mean{g_a\left(\lambda\left\|R_x\right\|\right)}=0
\end{equation*}
for all $g_a(y)=y^21_{\lbrace|y|>a\rbrace}$, $a>0$, which is trivial by $\mean{v^2}=1$.
\end{proof}

Notice that the covariances of the less isotropic ensemble of Proposition~\ref{Prop:scalingNIsotropic} can be read off from this proof. The main theorem will now follow from this lemma and the difference equation~(\ref{evolA}).
\begin{proof}[Proof of Theorem~\ref{thm:convA}]
We simplify the notation of (\ref{evolA}) and (\ref{defcaA}) by writing the real and imaginary parts of the matrix entries as elements of vectors in $\bbR^d$, $d=8N^2$. In this notation (\ref{evolA}) reads
\begin{equation}
\label{evolX}
 X^{\lambda}_j(y)-X^{\lambda}_j(y-1)=\lambda\sum_{k=1}^d\xi^{\lambda}_{jk}(y)X^{\lambda}_k(y-1)
\end{equation}
for all $y$ in $\bbN$, with $\xi^{\lambda}(y)\in\bbR^{d\times d}$ independent  of $X^{\lambda}(z), \xi^{\lambda}(z)$ $z\in\lbrace0,...,y-1\rbrace$. Because of the phase factors, the law of $\xi^{\lambda}(y)$  is not independent of $y\in\bbN$, but $\bbE(\xi^{\lambda}(y))=0$ and
\begin{equation} \label{eq: apriori bound xis}
 \|\xi^{\lambda}(y)\|^2  \leq  c'   (\sum_{z=1}^N \streep v(y,z)\streep)^2  \leq c  \sum_{z=1}^N v^2(y,z)
\end{equation}
where $c,c'< \infty$ are $\lambda$-independent for sufficiently small $\lambda$, but they depend on $N$.
This follows from (\ref{Rxsimplified}) or (\ref{Rxsimplifiedgamma0}) by noting that the velocity matrix $v^+$ converges to a nonsingular limit as $\lambda\rightarrow0$; we will henceforth assume without comment that $\lambda$ is sufficiently small. Furthermore, we know from the proof of Lemma \ref{lma_convergeZ}
\begin{equation}
\label{convxi}
 \lim_{\lambda\rightarrow0}\lambda^2\sum_{y=1}^{\left\lfloor\lambda^{-2}s\right\rfloor}\mean{\xi_{ik}^{\lambda}(y)\xi_{jl}^{\lambda}(y)}=C_{ikjl}\cdot s
\end{equation}
uniformly for $s\geq0$ from bounded intervals. 

If we define $\left(\caB(s)\right)_{s\geq0}$ as the matrix-valued Brownian motion with bracket process
\begin{equation*}
 \left\langle\caB_{ik},\caB_{jl}\right\rangle_s= C_{ikjl}\cdot s,
\end{equation*}
the equation (\ref{defcaA}) transforms to
\begin{equation*}
 \dd \caX_j(s)=\sum_{k=1}^d\dd\caB_{jk}\caX_k(s).
\end{equation*}
 The initial values for $X^{\lambda}$ and $\caX$ are identical and deterministic, the $\bbR^d$ vector corresponding to the unit matrix $1_{2N}$. For notational convenience, we have chosen $X^{\lambda}$ to still live on the microscopic, discrete space, what we really want to investigate is the cadlag process
\begin{equation*}
 \overline{X}^{\lambda}(s)=\rX\left(\left\lfloor\lambda^{-2}s\right\rfloor\right).
\end{equation*}
For cadlag processes we define $\overline{X}^{\lambda}(s-)$ as the leftside limit of $\overline{X}^{\lambda}$ at $s$. With the filtration $\caF^{\lambda}_s=\sigma\left\lbrace\overline{X}^{\lambda}(t):t\leq s\right\rbrace$, $\left(\oX(s)\right)_{s\geq0}$ is a $\left\lbrace\caF^{\lambda}_s \right\rbrace$-martingale. Furthermore defining
\begin{equation*}
 \rV_{ij}(y)=\lambda^2\sum_{x=1}^y\sum_{k,l=1}^d\mean{\xi_{ik}^{\lambda}(x)\xi_{jl}^{\lambda}(x)}\rX_k(x-1)\rX_l(x-1)
\end{equation*}
and the corresponding macroscopic
\begin{equation*}
 \oV_{ij}(s)=\rV\left(\left\lfloor\lambda^{-2}s\right\rfloor\right)
\end{equation*}
for all $i,j=1,...,d$, the process $\oX_i\oX_j-\oV_{ij}$ is a $\left\lbrace\caF^{\lambda}_s \right\rbrace$-martingale as well, and we have by Theorem 7.4.1, \cite{ethierkurtz},
\begin{lma}
\label{ek_lemma}
 If for any $T>0$, and any stopping time 
\begin{equation*}
 T\upl_r=\inf\left\lbrace s:\left|\oX(s)\right|\geq r \mbox{ or } \left|\oX(s-)\right|\geq r\right\rbrace,
\end{equation*}
$r>0$, for all $i,j=1,...,d$
\begin{align}
\label{estjump}
 &\lim_{\lambda\rightarrow0}\mean{\sup_{s\leq T\wedge T\upl_r}\left|\oX(s)-\oX(s-)\right|^2}=0\\
\label{brackjump}
&\lim_{\lambda\rightarrow0}\mean{\sup_{s\leq T\wedge T\upl_r}\left|\oV_{ij}(s)-\oV_{ij}(s-)\right|}=0
\end{align}
and
\begin{equation}
\label{ekbrackconv}
 \sup_{s\leq T\wedge T\upl_r}\left|\oV_{ij}(s)-\int_0^s\dd t\sum_{k,l=1}^dC_{ikjl}\oX_k(t)\oX_l(t)\right|\stackrel{\bbP}{\rightarrow}0,
\end{equation}
then $\left(\oX(s)\right)_{s\geq0}$ converges in distribution on $D_{\bbR^d}[0,\infty)$ to $\left(\caX(s)\right)_{s\geq0}$.
\end{lma}
\noindent{}So to prove Theorem \ref{thm:convA}, we only have to verify the conditions of Lemma \ref{ek_lemma}. 

%
%Since this had little to do with our specific problem, we present the verification in Appendix \ref{appendixA}. 
%This finishes the proof of (\ref{ekbrackconv}), and of Theorem \ref{thm:convA}.
%  
%As promised, we check the conditions of Lemma   \ref{ek_lemma}. 
 We start with the following observation,
\begin{lma}
\label{lma_Emax}
 \label{lma_L1}
Let $Z_k$, $k\in\bbN$ be  i.i.d. distributed, positive random variables, with
\begin{equation*}
\mean{Z_1}<\infty.
\end{equation*}
Then
\begin{equation*}
 \frac{1}{n}\mean{\max_{1\leq k\leq n}Z_k}\rightarrow0
\end{equation*}
as $n\rightarrow\infty$.
\end{lma}
\begin{proof}  
With
\begin{equation*}
 q(x)=\bbP\left(Z_1\geq x\right)
\end{equation*}
for $x\geq0$, we have
\begin{equation*}
 \int_0^\infty q(x)\dd x=\mean{Z_1},
\end{equation*}
while
\begin{equation*}
 \bbP\left(\max_{1\leq k\leq n}Z_k\geq x\right)=1-(1-q(x))^n\leq n q(x).
\end{equation*}
Thus,
\begin{equation*}
 \frac{1}{n}\mean{\max_{1\leq k\leq n}Z_k}=\int_0^\infty\frac{1-(1-q(x))^n}{n}\dd x
\end{equation*}
with the integrand on the right side converging to zero and dominated by the integrable $q(x)$, an the claim follows by dominated convergence.
\end{proof}
\noindent{}For any $T>0, r>0$ given, we have
\beq
\begin{split}
\label{proofjump}
\sup_{s\leq T\wedge T\upl_r}&\left|\oX(s)-\oX(s-)\right|^2\\
=&\max_{1\leq y\leq\left\lfloor\lambda^{-2} \left(T\wedge T\upl_r\right)\right\rfloor}\left|\rX(y)-\rX(y-1)\right|^2\\
=&\max_{1\leq y\leq\left\lfloor\lambda^{-2} \left(T\wedge T\upl_r\right)\right\rfloor}\lambda^2\sum_{j,k,l=1}^d\xi\upl_{jk}(y)\xi\upl_{jl}(y)\rX_{k}(y-1)\rX_{l}(y-1)\\
\leq&r^2\lambda^2\max_{1\leq y\leq\left\lfloor\lambda^{-2}T\right\rfloor}\left\|\xi\upl(y)\right\|^2\\
\leq&c(\lambda)r^2\lambda^2\max_{1\leq y\leq\left\lfloor\lambda^{-2}T\right\rfloor}\sum_{z=0}^{N-1}v(y,z)^2,
\end{split}
\eeq
where we used \eqref{eq: apriori bound xis} in the last line. Now the last line of (\ref{proofjump}) vanishes in expectation by Lemma \ref{lma_Emax} and the fact that $\bbE(v(y,z)^2)=1$. This proves (\ref{estjump}). For the proof of (\ref{brackjump}), note
\begin{equation*}
\begin{split}
 \sup_{s\leq T\wedge T\upl_r}&\left|\oV_{ij}(s)-\oV_{ij}(s-)\right|\\
=&\max_{1\leq y\leq\left\lfloor\lambda^{-2} \left(T\wedge T\upl_r\right)\right\rfloor}\left|\rV_{ij}(y)-\rV_{ij}(y-1)\right|\\
=&\lambda^2\max_{1\leq y\leq\left\lfloor\lambda^{-2} \left(T\wedge T\upl_r\right)\right\rfloor}\left|\sum_{k,l=1}^d\mean{\xi_{ik}^{\lambda}(y)\xi_{jl}^{\lambda}(y)}\rX_k(y-1)\rX_l(y-1)\right|\\
\leq&\lambda^2r^2\max_{1\leq y\leq\left\lfloor\lambda^{-2} T\right\rfloor}\mean{\left\|\xi\upl(y)\right\|^2},
\end{split}
\end{equation*}
which obviously vanishes in expectation as $\lambda\rightarrow0$.

For (\ref{ekbrackconv}), start with
\begin{equation}
\label{inttosum}
\begin{split}
 \oV_{ij}(s)&-\int_0^s\dd t\sum_{k,l=1}^dC_{i,k,j,l}\oX_k(t)\oX_l(t)\\
=&\sum_{k,l=1}^d\left(\lambda^2\sum_{y=1}^{\left\lfloor\lambda^{-2}s\right\rfloor}\left(\mean{\xi_{ik}^{\lambda}(y)\xi_{jl}^{\lambda}(y)}-C_{ikjl}\right)\rX_k(y-1)\rX_l(y-1)\right)\\
&\hspace{2cm}-\left(s-\lambda^2\left\lfloor\lambda^{-2}s\right\rfloor\right)\sum_{k,l=1}^dC_{ikjl}\rX_k\left(\left\lfloor\lambda^{-2}s\right\rfloor\right)\rX_l\left(\left\lfloor\lambda^{-2}s\right\rfloor\right).
\end{split}
\end{equation}
For the last line, we have
\begin{equation*}
\begin{split}
 \sup_{s\leq T\wedge T\upl_r}&\left|\left(s-\lambda^2\left\lfloor\lambda^{-2}s\right\rfloor\right)\sum_{k,l=1}^dC_{ikjl}\rX_k\left(\left\lfloor\lambda^{-2}s\right\rfloor\right)\rX_l\left(\left\lfloor\lambda^{-2}s\right\rfloor\right)\right|\\
&\leq\lambda^2r^2d^2\max_{k,l}\left|C_{ikjl}\right|\rightarrow0
\end{split}
\end{equation*}
almost surely, and thus in probability, as $\lambda\rightarrow0$. We omit the (finite) $k,l$ sum in (\ref{inttosum}) from our notation, and use partial summation with respect to $y$ to obtain
\begin{equation}
\label{partialsum}
 \begin{split}
  \lambda^2&\sum_{y=1}^{\left\lfloor\lambda^{-2}s\right\rfloor}\left(\mean{\xi_{ik}^{\lambda}(y)\xi_{jl}^{\lambda}(y)}-C_{ikjl}\right)\rX_k(y-1)\rX_l(y-1)\\
&=\lambda^2\rX_k\left(\left\lfloor\lambda^{-2}s\right\rfloor\right)\rX_l\left(\left\lfloor\lambda^{-2}s\right\rfloor\right)\sum_{y=1}^{\left\lfloor\lambda^{-2}s\right\rfloor}\left(\mean{\xi_{ik}^{\lambda}(y)\xi_{jl}^{\lambda}(y)}-C_{ikjl}\right)\\
&\hspace{2mm}-\lambda^2\sum_{y=1}^{\left\lfloor\lambda^{-2}s\right\rfloor}\sum_{x=1}^{y}\left(\mean{\xi_{ik}^{\lambda}(x)\xi_{jl}^{\lambda}(x)}-C_{ikjl}\right)\left(\rX_k(y)\rX_l(y)-\rX_k(y-1)\rX_l(y-1)\right).
 \end{split}
\end{equation}
We know from the convergences (\ref{convxi}) and
\begin{equation}
\label{convc}
 \lim_{\lambda\rightarrow0}\lambda^2\sum_{x=1}^{\left\lfloor\lambda^{-2}s\right\rfloor}C_{ikjl}=C_{ikjl}\cdot s,
\end{equation}
which are both uniform for $s$ from compact sets, that
\begin{equation*}
 a^{\lambda}_{ikjl}(y):=\lambda^2\sum_{x=1}^{y}\left(\mean{\xi_{ik}^{\lambda}(x)\xi_{jl}^{\lambda}(x)}-C_{ikjl}\right)\rightarrow0
\end{equation*}
as $\lambda\rightarrow0$ uniformly in $y$ as long as $1\leq y\leq\left\lfloor\lambda^{-2}T\right\rfloor$ for fixed positive $T$.

Thus for the first term on the right-hand side in (\ref{partialsum}),
\begin{equation*}
\begin{split}
 \sup_{s\leq T\wedge T\upl_r}&\left|\lambda^2\rX_k\left(\left\lfloor\lambda^{-2}s\right\rfloor\right)\rX_l\left(\left\lfloor\lambda^{-2}s\right\rfloor\right)\sum_{y=1}^{\left\lfloor\lambda^{-2}s\right\rfloor}\left(\mean{\xi_{ik}^{\lambda}(y)\xi_{jl}^{\lambda}(y)}-C_{ikjl}\right)\right|\\
&\leq r^2\sup_{s\leq T}\max_{i,k,j,l}\left|a^{\lambda}_{ikjl}\left(\left\lfloor\lambda^{-2}s\right\rfloor\right)\right|\rightarrow0
\end{split}
\end{equation*}
as $\lambda\rightarrow0$. After plugging (\ref{evolX}) into the second term on the right-hand side of (\ref{partialsum}), we are left with the sum of
\begin{equation*}
\begin{split}
 \sup_{s\leq T\wedge T\upl_r}\left|\lambda\sum_{y=1}^{\left\lfloor\lambda^{-2}s\right\rfloor}a^{\lambda}_{ikjl}(y)\left(\sum_{k'=1}^d \xi\upl_{kk'}(y)\rX_{k'}(y-1)\rX_l(y-1)\right.\right.\\
+\left.\left.\sum_{l'=1}^d \xi\upl_{ll'}(y)\rX_{k}(y-1)\rX_{l'}(y-1)\right)\vphantom{\sum_{y=1}^{\left\lfloor\lambda^{-2}s\right\rfloor}}\right|
\end{split}
\end{equation*}
and
\begin{equation*}
\begin{split}
 \sup_{s\leq T\wedge T\upl_r}\left|\lambda^2\sum_{y=1}^{\left\lfloor\lambda^{-2}s\right\rfloor}a^{\lambda}_{ikjl}(y)\left(\sum_{k',l'=1}^d \xi\upl_{kk'}(y)\xi\upl_{ll'}(y)\rX_{k'}(y-1)\rX_{l'}(y-1)\right)\right|\\
\end{split}
\end{equation*}
converging to zero in $L^2(\bbP)$ and $L^1(\bbP)$, respectively.

 \end{proof}

\section*{Acknowledgements}
Maximilian Butz benefited a lot from discussions with members of Antti Kupiainen's group at Helsinki University, and is grateful for financial support by the Academy of Finland during his stay there.  Sven Bachmann gratefully acknowledges the support
of the National Science Foundation under Grant \#DMS-0757581

%%%%%%%%%

%
% 
\renewcommand{\theequation}{A-\arabic{equation}}
  % redefine the command that creates the equation no.
  \setcounter{equation}{0}  % reset counter

\bibliography{Ref_DMPK}

\begin{thebibliography}{10}

\bibitem{BdR}
S.~Bachmann and W.~De~Roeck.
\newblock From the {Anderson} model on a strip to the {DMPK} equation and
  random matrix theory.
\newblock {\em J. Stat. Phys.}, 139:541--564, 2010.

\bibitem{beenakkerreview}
C.~W.~J. Beenakker.
\newblock Random-matrix theory of quantum transport.
\newblock {\em Rev. Mod. Phys.}, 69:731--808, 1997.

\bibitem{beenakkerrejaei}
C.~W.~J. Beenakker and B.~Rejaei.
\newblock Nonlogarithmic repulsion of transmission eigenvalues in a disordered
  wire.
\newblock {\em Phys. Rev. Lett.}, 71:3689--3692, 1993.

\bibitem{brouwernonuniversality}
P.W. Brouwer, C.~Mudry, and A.~Furusaki.
\newblock Nonuniversality in quantum wires with off-diagonal disorder: a
  geometric point of view.
\newblock {\em Nuclear Physics B}, 565(3):653 -- 663, 2000.

\bibitem{LB}
M.~B\"uttiker, Y.~Imry, R.~Landauer, and S.~Pinhas.
\newblock Generalized many-channel conductance formula with application to
  small rings.
\newblock {\em Phys. Rev. B}, 31:6207--6215, 1985.

\bibitem{Butzpaper}
M.~Butz.
\newblock {DMPK} eigenvalue process: Well-posedness and derivation from a
  random matrix model.
\newblock In preparation, 2012.

\bibitem{caselledistribution}
M.~Caselle.
\newblock Distribution of transmission eigenvalues in disordered wires.
\newblock {\em Phys. Rev. Lett.}, 74:2776 -- 2779, 1995.

\bibitem{TOETUE}
M.~Caselle.
\newblock A new classification scheme for random matrix theories.
\newblock arxiv.org/pdf/cond-mat/9610017, 1996.

\bibitem{D}
O.~N. Dorokhov.
\newblock Transmission coefficient and the localization length of an electron
  in {N} bound disordered chains.
\newblock {\em JETP Lett.}, 36(7):318--321, 1982.

\bibitem{ErYa12}
L.~{Erd\"o}s and H.T. Yau.
\newblock Universality of local spectral statistics of random matrices.
\newblock {\em Bull. Amer. Math. Soc.}, January 2012.

\bibitem{ethierkurtz}
Stewart~N. Ethier and Thomas~G. Kurtz.
\newblock {\em Invariance Principles and Diffusion Approximations}, pages
  337--364.
\newblock {John Wiley \& Sons, Inc.}, 2008.

\bibitem{huffmann}
A.~{H\"uffmann}.
\newblock Disordered wires from a geometric viewpoint.
\newblock {\em J. Phys. A}, 23(24):5733, 1990.

\bibitem{js}
J.~Jacod and A~Shiryaev.
\newblock {\em Limit Theorems for Stochastic Processes}.
\newblock Springer Verlag, 2nd edition, 2003.

\bibitem{LeSt85}
P.A. Lee and A.D. Stone.
\newblock Universal conductance fluctuations in metals.
\newblock {\em Phys. Rev. Lett.}, 55:1622--1625, 1985.

\bibitem{chalkermacedo}
A.M.S Macedo and J.~Chalker.
\newblock Exact results for the level density and two-point correlation
  function of the transmission-matrix eigenvalues in quasi-one-dimensional
  conductors.
\newblock {\em Phys. Rev. B.}, 49(7):4695 -- 4702, 1994.

\bibitem{MPK}
P.~A. Mello, P.~Pereyra, and N.~Kumar.
\newblock Macroscopic approach to multichannel disordered conductors.
\newblock {\em Annals of Physics}, 181(2):290--317, 1988.

\bibitem{mellostone}
P.A. Mello and A.D. Stone.
\newblock Maximum-entropy model for quantum-mechanical interference effects in
  metallic conductors.
\newblock {\em Phys. Rev. B}, 44:3559--3576, 1991.

\bibitem{RoSc10}
R.~R{\"o}mer and H.~Schulz-Baldes.
\newblock The random phase property and the {Lyapunov} spectrum for disordered
  multi-channel systems.
\newblock {\em J. Stat. Phys.}, 140:122--153, 2010.

\bibitem{SaSc10}
C.~Sadel and H.~Schulz-Baldes.
\newblock Random {Lie} group actions on compact manifolds: A perturbative
  analysis.
\newblock {\em Ann. Prob.}, 38(6):2224--2257, 2010.

\bibitem{viragvalkoexplosion}
B.~Valko and B.~Virag.
\newblock Random {Schr\"odinger} operators on long boxes, noise explosion and
  the {GOE}.
\newblock http://arxiv.org/abs/0912.0097v3, 2009.

\bibitem{zirnbauersuperfourier}
M.~Zirnbauer.
\newblock Super {F}ourier analysis and localization in disordered wires.
\newblock {\em Phys. Rev. Lett.}, 69:1584 -- 1587, 1992.

\end{thebibliography}
\bibliographystyle{plain}

\end{document}